\documentclass{article}

\usepackage{PRIMEarxiv}
\usepackage{authblk}
\usepackage[utf8]{inputenc} 
\usepackage[T1]{fontenc}    
\usepackage{hyperref}       
\usepackage{url}            
\usepackage{booktabs}       
\usepackage{amsfonts}       
\usepackage{nicefrac}       
\usepackage{microtype}      
\usepackage{lipsum}
\usepackage{fancyhdr}       
\usepackage{graphicx}       
\graphicspath{{media/}}     
\usepackage{xspace}
\usepackage{stackengine}
\usepackage{multicol}
\usepackage{bm}
\usepackage{siunitx}

\usepackage{multirow}
\usepackage{breakcites}
\usepackage{float}
\usepackage{amsmath}
\usepackage{mathtools} 
\usepackage{booktabs} 
\usepackage{tikz} 
\usepackage{graphicx} 
\usepackage{subfig}
\usepackage{xcolor}
\usepackage{hyperref}
\usepackage{amsthm}

\newtheorem{dfn}{Definition}
\newtheorem{thm}{Theorem}
\newtheorem{lem}{Lemma}

\newcommand{\improvement}{improvement\xspace} 
\newcommand{\equitability}{shortfall\xspace} 
\newcommand{\gainovermin}{gain\xspace}
\newcommand{\regret}{regret\xspace}

\newcommand{\Improvement}{Improvement\xspace} 
\newcommand{\Equitability}{Shortfall\xspace} 
\newcommand{\Gainovermin}{Gain\xspace}
\newcommand{\Regret}{Regret\xspace}
\title{Trade-offs between Group Fairness Metrics in Societal Resource Allocation}
\author[1]{Tasfia Mashiat}
\author[1]{Xavier Gitiaux}
\author[1]{Huzefa Rangwala}
\author[2]{Patrick Fowler}
\author[1]{Sanmay Das}
\affil[1]{George Mason University. Email:\{tmashiat, xgitiaux, rangwala, sanmay\}@gmu.edu}
\affil[2]{Washington University in Saint Louis. Email: pjfowler@wustl.edu}

\begin{document}

\maketitle

\begin{abstract}
We consider social resource allocations that deliver an array of scarce supports to a diverse population. Such allocations pervade social service delivery, such as provision of homeless services, assignment of refugees to cities, among others. At issue is whether allocations are fair across sociodemographic groups and intersectional identities.
Our paper shows that necessary trade-offs exist for fairness in the context of scarcity; many reasonable definitions of equitable outcomes cannot hold simultaneously except under stringent conditions. For example, defining fairness in terms of improvement over a baseline inherently conflicts with defining fairness in terms of loss compared with the best possible outcome.
Moreover, we demonstrate that the fairness trade-offs stem from heterogeneity across groups in intervention responses. 
Administrative records on homeless service delivery offer a real-world example. Building on prior work, we measure utilities for each household as the probability of reentry into homeless services if given three homeless services. Heterogeneity in utility distributions (conditional on received services) for several sociodemographic groups (e.g. single women with children versus without children) generates divergence across fairness metrics. We argue that such heterogeneity, and thus, fairness trade-offs pervade many social policy contexts. 

\end{abstract}


\section{Introduction}\label{sec:intro}

Many social interventions that allocate resources to individuals are challenging because individuals have heterogeneous utilities. 
Thus, the design and analysis of allocation policies for social interventions in terms of efficiency and fairness is critical \cite{roth2015gets}, as seen in many domains including child protection (e.g. \cite{chouldechova2018case}), healthcare (e.g. \cite{yadav2016using}), and homeless services (e.g \cite{kube2019fair,brown2018reliability}). 
A particular concern for the use of machine learning posits that the tools systematically disfavor some sociodemographic or intersectional groups (see \cite{chouldechova2018frontiers} for a review). For example, a growing body of work has documented racial disparities in credit lending, recidivism risk assessment \cite{ProPublica2016}, education \cite{gardner2019evaluating}, healthcare \cite{pfohl2019creating}, and policing \cite{ensign2018runaway}. In this paper, we explore how to measure these potential disparities in the context of allocating resources given a limited budget. The literature on fair resource allocation has typically come from the areas of fair division and cooperative game theory. In that literature, one typically thinks of individuals as having preferences, and tries to define measures of fairness and  allocation mechanisms that demonstrate these properties with respect to individual preferences. Recent notions of group fairness coming from the fair division line of literature strengthen the requirements for individual fairness \cite{conitzer2019group} and are thus too strong for situations of scarce resource allocation, where allocations by definition must be unfavorable to some individuals.

So how should one measure fairness across groups in the allocation of scarce societal resources, where decisions often are made on the basis of multiple criteria? To ground our considerations in a specific case, consider homelessness service provision, where federal policy makes serving the most vulnerable an explicit goal, and at the same time, the effectiveness of services is measured by returns to homelessness among those served \cite{systemperformance}. Such examples motivate us to consider how different notions of what role social services should play lead to different conclusions about the fairness of potential allocations across demographic groups. 

For example, we could analyze how much better off members of a group are compared with how well they would have done under some minimal baseline allocation, or we could look at how much worse-off members of a group are than they would have been under the allocations that serve them the best.
Fairness could then be defined as equitable performance of groups according to these measures, and indeed, the existing literature on fair allocation of both divisible and indivisible resources has looked at measures along both of directions, of \improvement (or \gainovermin) and \regret (or \equitability). 
Although both are reasonable definitions of a fair allocation, we consider two important factors that arise in many real-world problems. First, instead of the problem simply focusing on a set of identical resources that need to be allocated amongst agents, there is often a whole set of different interventions, each with capacity constraints (for example, different types of homelessness resources or different cities that refugees can be matched to). Second, individuals may respond heterogeneously to the different interventions (for example, homeless individuals with disabilities may benefit disproportionately from intensive housing supports, or refugees may assimilate and find jobs more easily in places where there is already a substantial population from their place of origin). 

We show that when there is a multiplicity of possible services
, and groups are heterogeneous in the distributions of utilities they receive from different services, it becomes impossible to satisfy simultaneously \improvement and \regret oriented definitions of group fairness. Even more dramatically, 
an allocation policy that appears to favor one group according to \improvement fairness can favor another group according to \regret fairness. The results yield insights on inherent trade-offs that policymakers face when attempting to achieve a fairness objective. How we measure improvement or regret also matters when assessing the fairness of an allocation policy. For example, we could measure improvement by the ratio of realized utility over baseline utility (a multiplicative measure); or by the difference between realized utility and baseline utility (an additive measure). Depending on the application, it is not always clear which of these additive or multiplicative normalizations makes more sense. We establish, in a stylized framework, that fairness in terms of additive normalization and fairness in terms of multiplicative normalization cannot hold simultaneously except when the distribution of individual responses to different allocations is similar across demographic groups. 

These trade-offs are not theoretical corner-cases and have substantive implications for social policy. We use administrative data from a regional homeless system to explore the fairness of a capacitated assignment of community-based services that address housing needs. Services include transitional housing, rapid rehousing, and emergency shelter; three programs that vary in intensity and availability. We measure the utility of a service to a household as the probability estimated in prior work by \cite{kube2019fair} that the household would make a successful exit from homelessness given the delivery of that service. We first document significant differences in utility distributions across different groups (e.g., disabled versus not disabled households, families with children versus households without children, single females with versus without children). We then confirm our theoretical results that the differences in utility distributions across groups generate trade-offs when assessing the fairness of an allocation. For example, we consider the original allocation as recorded in the administrative data and we find that improvement and regret disagree on whether the policy favors households with or without children, as well as other groups. 

In addition to contributing to our understanding of how the definition and measurement of fairness is affected by heterogeneity in how members of different groups may respond to interventions, these findings can inform practice in homeless and social services that allocate scarce resources across diverse populations. Policies frequently attempt to maximize public welfare by targeting available supports towards heterogeneous groups based on competing notions of fairness (e.g., vulnerability, efficiency, equality). Understanding the fairness trade-offs and measurement sensitivity allows for more intentional policy-making and better evaluation.

\section{Related Work}\label{sec:literature}

\subsection{Group Fairness}
Prior research has led to many definitions of fairness to compare algorithmic outcomes across demographic groups. Popular definitions include statistical parity~\cite{dwork2012fairness}, equalized odds and opportunity~\cite{hardt2016equality}. However, these definitions only apply to binary settings and implicitly assume that the utility of an individual is equal to one when the algorithm's outcome is one and equal to zero otherwise. Few papers consider more general definitions of utilities~\cite{heidari2019moral}. In this paper, we argue as in~\cite{hossain2020designing} that in many societal applications of machine learning, utilities are heterogeneous across individuals and that this heterogeneity could be systematic across demographic groups. 

The fair division literature offers a framework to compare utilities across individuals. Envy-freeness, proportionality or equitability~\cite{caragiannis2019unreasonable} are common utility-based definitions of a fair allocation of goods. The literature strengthens these notions of fairness to control for envy-freeness to arbitrary segments of the population~\cite{conitzer2019group, bartholdi1992hard}. In this paper, we focus on notions of group equitability that vary by their normalization, but leaves it for future research to explore the role of normalization on group envy-freeness. 

A standard assumption in the fair division literature is that utilities, although heterogeneous, are unit-normalized~\cite{aziz2020justifications}. The rationale for unit-normalization is that it allows one to make more reasonable interpersonal comparisons of utility by converting all utilities to a common scale. Unit-normalization implies that the maximum utility gain is equal to one for all individuals~\cite{aziz2020justifications}. Our notions of \equitability or \regret rely on a similar assumption, which is reasonable in many settings (e.g. voting ~\cite{bouveret2016characterizing}). However, we argue that other reasonable choices of normalization are possible and more relevant in different types of allocation problems. For example, in the case of homeless services delivery, a policymaker would want to account for the fact that families with children have on average more to gain from rapid rehousing programs \cite{rog2007characteristics}. In this case, our measures of \improvement and \gainovermin, which normalize by comparison with the worst utility that an individual can expect from an allocation, are also reasonable notions of fairness.  This paper relates closely to the work of \cite{hossain2020designing}, who introduce utility-based notions of group fairness for classification problems. However, they assume away the need to normalize utilities to a similar scale / support. One of our contributions is to show that different normalization approaches can lead to conflicting assessments of the fairness of an allocation policy.

\subsection{Impossibility Results}
The binary outcome setting admits some fundamental impossibility results  \cite{kleinberg2016inherent,chouldechova2017fair}. Except under very restrictive conditions, it is impossible for a classifier to simultaneously equalize false positive rates and false negative rates across groups and also guarantee that predictions are calibrated within each group. \cite{kleinberg2016inherent} show that the impossibility emerges whenever demographic groups differ systematically in the distribution of features used by the classifier as inputs. In this paper, we demonstrate new impossibility results in the case of utility-based notions of fairness. As in \cite{kleinberg2016inherent}, we obtain a paradox where fairness guarantees that seem to share the same objective -- that the allocation of resources will be as effective for all demographic groups -- are nonetheless incompatible.  

Our results on the incompatibility of different fairness principles is also reminiscent of Arrow's impossibility theorem ~\cite{arrow1950difficulty}. In the presence of heterogeneous preferences, there is no way to aggregate individual preferences into a social welfare function that would satisfy unanimity, non-dictatorship and informational parsimony. The theory of fair allocation \cite{foley1966resource,varian1973equity} that selects a subset of policies on basis of their fairness and efficiency obtains possibility results by relaxing informational parsimony \cite{fleurbaey2005informational}. However, in this paper, we show that we cannot avoid negative results when notions of fairness based on different normalizations have to hold simultaneously.

\subsection{Algorithmic Allocation of Societal Resources}
There has been 
recent interest in the specific setting where scarce resources that are collective or societal are algorithmically allocated by a centralized institution to individual members of society (see \cite{das2022local} for a recent review). 
The design of algorithmic approaches has typically focused on increasing the efficiency  of social interventions, including kidney exchange \cite{li2019incorporating,roth2005pairwise}, housing assistance \cite{kube2019fair,manlove2006popular}, HIV awareness campaigns \cite{yadav2016using} and refugee resettlement \cite{delacretaz2019matching}. In this paper, we investigate how to assess the fairness of resulting allocations. Empirically, we find evidence of our impossibility results in the context of capacitated one-sided matching, which involve a set of services with fixed capacities, a set of agents with heterogeneous preference orderings (see e.g. ~\cite{manlove2006popular} for an application to the house allocation problem) and a social worker that assigns a service to each agent.

\section{Inherent Fairness Trade-Offs in Resource Allocation}
\label{sec: method}

In this section we describe our theoretical framework, first defining the problems we are concerned with, and then outlining both general and illustrative results on inherent group fairness trade-offs in the allocation of scarce resources.

\subsection{Setting}
We consider $K$ services, with maximum capacities $c_{k}$ for $k\in \{1, ..., K\}$, and $N$ individuals $i=1, ..., N$.\footnote{We follow the convention of denoting vectors in bold type and random variables with capital letters.} We can thus describe individuals by their utility vector $\mathbf{u}=(u_{1}, ..., u_{K})$ over each program $k$ and their sensitive attribute $s\in\mathcal{S}$. $\mathcal{S}$ describes the set of groups for which we want to study the fairness of service allocation. For ease of exposition, we assume that group characteristics are binary and $\mathcal{S}=\{0, 1\}$; however, our results readily extend to more complex definitions of groups, and the empirical section will show that our results hold for intersectional groups. We denote by $N_{s}$ the number of individuals with sensitive attribute $s=0, 1$.

For each individual $\mathbf{u}$, we denote by $u^{\min}$ the utility derived from receiving the least beneficial program: $u^{\min}=\min\{u_{k}|k=1, .., K\}$. We denote by $u^{\max}$ the utility derived from receiving the most beneficial program: $u^{\max}=\max\{u_{k}|k=1, .., K\}$. Best and worst programs might vary among individuals. $u^{\min}$ could potentially characterize a ``do nothing option'', i.e. the individuals' utility without the intervention. We assume that $\mathbf{u}$ is drawn from a joint distribution $G_{s}(u)$ over $\mathbb{R}^{K}$ that depends on the value $s$ of the sensitive attribute. We denote the random utility vector $U$.

An allocation policy $\mathbf{a}:\mathbb{R}^{K}\rightarrow \{0, 1\}^{K}$ assigns each individual with utility $\mathbf{u}$ to a program $k$ if and only if  $a_{k}(\mathbf{u})=1$. We assume that individuals are assigned to only one program: $\sum_{k=1}^{K}a_{k}(\mathbf{u})=1$. We denote by $\mathbf{a}.\mathbf{u}$ the inner product between the policy assignment and the individual utility: $\mathbf{a}.\mathbf{u} = \sum_{k=1}^{K}a_{k}(\mathbf{u})u_{k}$. Given $N$ individuals $i$ with utility $\mathbf{u}_{i}$, the allocation is feasible if and only if for all programs $k$, $\sum_{i=1}^{N}a_{k}(\mathbf{u}_{i})\leq c_{k}$ (the maximum capacity for the $k$-th service). 

\subsection{Fairness, Baselines, and Normalization}
In this section, we consider four notions of fairness to compare the average realized utility between groups: \improvement, \regret, \equitability, and \gainovermin. The definitions differ along two dimensions (1) how they normalize individual utility (additive or multiplicative), and (2) which baselines they compare individual realized utility to (worst case or best case).

The \improvement and \gainovermin metrics use as a baseline the minimal or worst utility that an individual can expect from any service they receive. To be fair, the definitions say that the average increase in utility relative to the least beneficial intervention should be equal across groups. They differ in how they normalize realized utility relative to the baseline; \improvement uses an additive normalization, while \gainovermin uses a multiplicative normalization.

\begin{dfn}
\textbf{\Improvement fairness.}
An allocation policy $\mathbf{a}$ satisfies fair \improvement if and only if
\begin{equation}
    E\left[\frac{1}{N_{0}}\displaystyle\sum_{i, s=0}\mathbf{a}.(\mathbf{u}_{i} - u^{\min}_{i})\right]= E\left[\frac{1}{N_{1}}\displaystyle\sum_{i, s=1}\mathbf{a}.(\mathbf{u}_{i} - u^{\min}_{i})\right],
\end{equation}
where the expectation is taken over samples of size $N_{s}$ for the group with sensitive attribute $s=0,1$.
\end{dfn}

\begin{dfn}
\textbf{\Gainovermin fairness.}
An allocation policy $a$ satisfies fair \gainovermin if and only if
\begin{equation}
    E\left[\frac{1}{N_{0}}\displaystyle\sum_{i, s=0}\mathbf{a}.\frac{\mathbf{u}_{i}}{ u^{\min}_{i}}\right]= E\left[\frac{1}{N_{1}}\displaystyle\sum_{i, s=1}\mathbf{a}.\frac{\mathbf{u}_{i}}{ u^{\min}_{i}}\right].
\end{equation}
\end{dfn}

We denote by $\Delta I(\mathbf{a})$ the difference in \improvement between groups:
\begin{equation}
    \Delta I(\mathbf{a}) =  E\left[\frac{1}{N_{1}}\displaystyle\sum_{i, s=1}\mathbf{a}.(\mathbf{u}_{i} - u^{\min}_{i})\right] - E\left[\frac{1}{N_{0}}\displaystyle\sum_{i, s=0}\mathbf{a}.(\mathbf{u}_{i} - u^{\min}_{i})\right].
\end{equation}
If $\Delta I(\mathbf{a})$ is positive, the policy $\mathbf{a}$ favors group $1$; if $\Delta I(\mathbf{a})$ is negative, the policy favors group $0$. We define similarly differences in \gainovermin as $\Delta G(\mathbf{a})$.

\Regret fairness and \equitability benchmark the realized utility in comparison to the best outcome individuals can hope for from any service (as such they are related to the classical definition of \emph{equitability} in fair division, albeit with differences in normalization). Both fairness notions are satisfied when the average loss of utility compared to receiving the most beneficial program is equalized across groups. 

\begin{dfn}
\textbf{\Regret fairness.}
An allocation policy $a$ satisfies \regret fairness if and only if
\begin{equation}
    E\left[\frac{1}{N_{0}}\displaystyle\sum_{i: s=0}\mathbf{a}.(u^{\max}_{i}- \mathbf{u}_{i})\right]= E\left[\frac{1}{N_{1}}\displaystyle\sum_{i: s=1}\mathbf{a}.(u^{\max}_{i}- \mathbf{u}_{i})\right],
\end{equation}
\end{dfn}

\begin{dfn}
\textbf{\Equitability.}
An allocation policy $a$ satisfies \equitability if and only if
\begin{equation}
    E\left[\frac{1}{N_{0}}\displaystyle\sum_{i: s=0}\mathbf{a}.\frac{\mathbf{u}_{i}}{ u^{\max}_{i}}\right]= E\left[\frac{1}{N_{1}}\displaystyle\sum_{i: s=1}\mathbf{a}.\frac{\mathbf{u}_{i}}{ u^{\max}_{i}}\right].
\end{equation}
\end{dfn}

Like differences in \improvement or in \gainovermin, we denote differences in \equitability and \regret as $\Delta S(\mathbf{a})$ and $\Delta R(\mathbf{a})$, respectively. Note that $\Delta R(\mathbf{a})\geq 0$ means that the policy $\mathbf{a}$ favors group $S=0$ over group $S=1$ for regret fairness.

All four definitions represent reasonable and desirable properties of a fair allocation. However, the following results show that a decision-maker faces trade-offs when choosing which fairness notion to target. Not only might the notions not be satisfied simultaneously, it is possible to generate explicitly contradictory conclusions across the relatively similar fairness metrics regarding which group is under-served. 

\subsection{\Improvement versus \Regret}\label{subsec:imp_vs_regret}

Our first result shows that \improvement and \regret fairness cannot be satisfied simultaneously, unless we impose strong restrictions on how groups differ. Consider two random variables $U^{\max}$ and $U^{\min}$ defined on individual most and least beneficial utility. The maximum individual utility gain that can be delivered by a service is then a random variable $\Delta U = U^{\max}-U^{\min}$. We show that heterogeneity in $\Delta U$ across groups generates an inherent trade-off between improvement and regret fairness. 

\begin{thm}
\label{thm: trade-off}
If an allocation policy $\mathbf{a}$ satisfies both \improvement and \regret fairness then the average maximum utility gain $\Delta U$ must be equal across groups: $E[\Delta U|S=0] = E[\Delta U|S=1]$. Moreover, $\Delta I(\mathbf{a}) + \Delta R(\mathbf{a})=E[\Delta U|S=1] - E[\Delta U|S=0]$. 
\end{thm}

\begin{proof}
The proof is based on the following identities:
\begin{equation}
    \begin{split}
        \Delta I(\mathbf{a})& = E\left[\frac{1}{N_{1}}\displaystyle\sum_{i: s=1}\mathbf{a}(\mathbf{u}).(\mathbf{u}_{i}-u^{\max}_{i} + \Delta u_{i})\right] - E\left[\frac{1}{N_{0}}\displaystyle\sum_{i: s=0}\mathbf{a}(\mathbf{u}).(\mathbf{u}_{i}-u^{\max}_{i} + \Delta u_{i})\right]  \\
        & =E\left[\frac{1}{N_{1}}\displaystyle\sum_{i: s=1}\sum_{k=1}^{K}\mathbf{a}_{k}(\mathbf{u})\Delta u_{i}\right] - E\left[\frac{1}{N_{0}}\displaystyle\sum_{i: s=0}\sum_{k=1}^{K}\mathbf{a}_{k}(\mathbf{u})\Delta u_{i}\right] - \Delta R(\mathbf{a})\\
        & =E\left[\frac{1}{N_{1}}\displaystyle\sum_{i: s=1}\Delta u_{i}\right] - E\left[\frac{1}{N_{0}}\displaystyle\sum_{i: s=0}\Delta u_{i}\right] - \Delta R(\mathbf{a}),
    \end{split}
\end{equation}
where the last equality comes from the fact that $\sum_{k=1}^{K}a_{k}(u)=1$ for all $u$. Therefore, if $\Delta I(\mathbf{a})=0$ and $\Delta R(\mathbf{a})=0$, then $E[\Delta U|S=0] = E[\Delta U|S=1]$.
\end{proof}

The result in Theorem \ref{thm: trade-off} implies that regardless of the allocation policy, for both \improvement and \regret fairness to hold it is necessary that groups would gain on average similarly if they were always allocated their most beneficial intervention. Thus, a trade-off exists when defining what a fair assignment should look like: for example, a policy satisfying improvement fairness would always violate regret fairness unless $E[\Delta U|S=0] = E[\Delta U|S=1]$. Since $\Delta I(\mathbf{a}) + \Delta R(\mathbf{a})=E[\Delta U|S=1] - E[\Delta U|S=0]$, the closer a policy is to satisfying improvement fairness, the worse its regret fairness, and vice-versa. 
A follow up question is whether \improvement and \regret fairness tell a different story about the fairness of an allocation policy $a$. The next result shows that whenever $E[\Delta U|S=0]$ and  $E[\Delta U|S=1]$ differ, unless all policies favor one group, there exists a policy that favors one group for \improvement fairness and favors the other one for \regret fairness. 

\begin{thm}
\label{cor: 2}
Suppose that $E[\Delta U|S=1] > E[\Delta U | S=0]$. Suppose that there exists a policy that favors group $S=0$ for \improvement fairness and another policy that favors group $S=1$ for \improvement fairness. Then, there exists a policy $\mathbf{a^{*}}$ such that $\Delta I(\mathbf{a^{*}}) > 0 \mbox{ and } \Delta R(\mathbf{a^{*}}) > 0
$. That is, there exists a policy that favors $S=1$ with respect to \improvement fairness (larger is better), but favors $S=0$ with respect to \regret fairness (lower is better).
\end{thm}

The proof of Theorem \ref{cor: 2} relies on the fact that the set of differences in \improvement/\regret is a continuous interval:
\begin{lem}
\label{lem: 1}
Suppose that there exist two allocation policies $\mathbf{a}$ and $\mathbf{a}^{'}$ with differences in \improvement $\delta$ and $\delta^{'}> \delta$. Then, for any $\delta^{*}\in[\delta, \delta^{'}]$, there exists an allocation policy $\mathbf{a}^{*}$ with difference in \improvement equal to $\delta^{*}$. A similar result holds for difference in \regret. \end{lem}

\begin{proof}
We show the result for differences in \improvement. The proof can be readily extended to differences in \regret. We choose $\lambda = \frac{\delta^{'} - \delta^{*}}{\delta^{'} - \delta}\in [0, 1]$. We define an allocation policy $\mathbf{a}^{\lambda}$ as follows:
\begin{itemize}
    \item Partition randomly the individuals into two populations $G_{\lambda}$ and $G_{1-\lambda}$ of size $\lambda N$ and $(1-\lambda )N$, respectively.
    \item For each program $k$, assign $\lambda c_{k}$ of them to the population $G_{\lambda}$; and $(1-\lambda)c_{k}$ of them to the population $G_{1 -\lambda}$. 
    \item Apply the allocation policy $\mathbf{a}$ to the population $G_{\lambda}$ and $\mathbf{a}^{'}$ to the population $G_{1-\lambda}$. 
\end{itemize}
By construction the policy $a_{\lambda}$ satisfies the resource constraints. Moreover,
\begin{equation}
        \Delta I(\mathbf{a}_{\lambda}) = \Delta I(\mathbf{a}) P(G_{\lambda}) + \Delta I(\mathbf{a}^{'}) (1 - P(G_{1-\lambda}) =\delta \lambda + \delta^{'}(1-\lambda) = \delta^{*},
\end{equation}
where the last equality comes from our choice for the value of $\lambda$. 
\end{proof}

\begin{proof}
Theorem \ref{cor: 2}. We choose $\epsilon = \frac{E[\Delta U|S=1] - E[\Delta U | S=0]}{2}$. $\epsilon >0$ by assumption. Using the assumption of Theorem \ref{cor: 2}, there exist $\mathbf{a}$ and $\mathbf{a}^{'}$ such that $\Delta I(\mathbf{a}) < 0$ and $\Delta I(\mathbf{a}^{'}) > 0$. We apply Lemma \ref{lem: 1} with $\delta =\Delta I(\mathbf{a})< 0 $, $\delta^{'}=\Delta I(\mathbf{a}^{'}) > 0$ and $\delta^{*}=\min\{\epsilon, \delta^{'} / 2\}$: there exists a policy $\mathbf{a}^{*}$ such that $\Delta I(\mathbf{a}^{*}) = \delta^{*} > 0$. Moreover, $\Delta R(\mathbf{a}^{*}) = E[\Delta U|S=1] - E[\Delta U | S=0] - \Delta I(\mathbf{a}^{*}) \geq \epsilon > 0$.
\end{proof}

Thus, \regret fairness and \improvement fairness cannot hold simultaneously unless populations are homogeneous in terms of their best response to the allocation (Theorem \ref{thm: trade-off}). Moreover, assessing which group is favored by a given policy can lead to contradictory results depending on whether we measure the fairness properties of the policy in terms of differences in \improvement or regret. The result in Theorem \ref{cor: 2} illustrates that decision-makers cannot expect that \improvement and \regret notions tell a similar story about whether an allocation policy under-serves a given group. Results Theorem \ref{thm: trade-off} and Theorem \ref{cor: 2} are general, since they hold for any set of capacities $c_{1}$, ..., $c_{K}$ and for distributions of utilities provided that $E[\Delta U|S=1] > E[\Delta U | S=0]$. Both illustrate the central role of the difference between $E[\Delta U|S=0]$ and $E[\Delta U|S=1]$ in driving wedges between \improvement and \regret fairness. Additionally, Theorem \ref{cor: 2} is not very restrictive in its assumptions, since it only requires that no group is under-served regardless of the policy. 

\subsection{\Equitability versus \Gainovermin}
In this section, we show that the fairness trade-offs between \improvement and \regret exist also with multiplicative notion of fairness, \gainovermin and \equitability. Unlike trade-offs between improvement and regret where our results are general, in the case of \equitability versus \gainovermin, we derive results in a stylized framework and leave it to future work to extend our results to more general settings. Nevertheless, this section captures the essence of the problem in the multiplicative setting. We denote for each individual by $r=u^{\min}/u^{\max}$ the ratio between the lowest and highest utility obtained from the intervention. This serves as a multiplicative counterpart of $\Delta u$. We consider the following framework (SF1):
\begin{itemize}
    \item There are two types of individuals: type $A$ with high value $\overline{r}$ for the ratio $r$; type $B$ with a low value $\underline{r}< \overline{r}$ for $r$. 
    \item Conditional on $r$, the distribution of utility is similar across programs and types.
\end{itemize}
In this stylized framework, assigning to an individual their most beneficial program delivers either a large increase $\overline{r}$ over $u^{\min}$ (type A) or a lower one $\underline{r}$ (type B). We characterize the heterogeneity across groups by differences in the distribution of type A and B within each group. We denote by $\pi_{0}$ the proportion of type B individuals for group $S=0$; and, $\pi_{1}$ the proportion of type $B$ for group $S=1$.

\begin{thm}
\label{thm: equi}
In the stylized framework (SF1):
\begin{itemize}
\item A policy satisfies both \equitability and \gainovermin fairness if and only if  $\pi_{0}=\pi_{1}$. 
\item If $\pi_{0}\neq \pi_{1}$, a policy $a$ that achieves \gainovermin (\equitability) fairness, favors, according to \equitability (\gainovermin) fairness whichever group has the lowest proportion of type $A$ individuals.
\end{itemize}
\end{thm}

\begin{proof}
Let $\underline{\alpha}$  denote $E\left[\frac{\mathbf{a}(u).\mathbf{u}}{u^{\min}}|r=\underline{r}\right]$ and $\overline{\alpha}$  denote $E\left[\frac{\mathbf{a}(u).\mathbf{u}}{u^{\min}}|r=\overline{r}\right]$. Then, we write (for any policy) differences in \gainovermin as 
\begin{equation}
\label{eq: if}
    \Delta G(\mathbf{a}) = \left\{\pi_{1}\underline{\alpha} + (1 - \pi_{1}) \overline{\alpha}\right\}  -  \left\{\pi_{0}\underline{\alpha} + (1 - \pi_{0}) \overline{\alpha}\right\} =  (\pi_{0} - \pi_{1})(\overline{\alpha} - \underline{\alpha})
\end{equation}
and differences in \equitability as 
\begin{equation}
    \Delta S(\mathbf{a}) =\left\{\pi_{1}\underline{\alpha}\underline{r} + (1 - \pi_{1}) \overline{\alpha}\overline{r}\right\} -  \left\{\pi_{0}\underline{\alpha}\underline{r} + (1 - \pi_{0}) \overline{\alpha}\overline{r}\right\} = (\pi_{0} - \pi_{1})(\overline{\alpha}\overline{r} - \underline{\alpha}\underline{r}).
\end{equation}
Therefore, \gainovermin and \equitability fairness are equivalent to 
$(\pi_{0} - \pi_{1})(\overline{\alpha} - \underline{\alpha})=0$ and $(\pi_{0} - \pi_{1})(\overline{\alpha}\overline{r} - \underline{\alpha}\underline{r})=0$. Hence, if $\pi_{0}\neq \pi_{1}$, $\overline{\alpha}=\underline{\alpha}$ and $\overline{\alpha}\;\overline{r} = \underline{\alpha}\;\underline{r}$, which is not possible since $\underline{r} \neq \overline{r}$.  

To show the second part of Theorem \ref{thm: equi}, we use the fact that \gainovermin fairness implies that $\overline{\alpha}=\underline{\alpha}$ (equation \eqref{eq: if}) and that the difference in \equitability between group $S=1$ and $S=0$ can be then written $\Delta S(\mathbf{a})=(\pi_{0} - \pi_{1})(\overline{r}-\underline{r})\overline{\alpha}$, which have the same sign as $\pi_{0} - \pi_{1}$ since $\overline{r}>\underline{r}$. Therefore, if $\pi_{0} > \pi_{1}$, the policy favors group $S=1$ with respect to \equitability fairness; otherwise, it favors group $S=0$.
\end{proof}

Theorem \ref{thm: equi} states that \equitability and \gainovermin can be satisfied simultaneously if and only if populations have similar fractions of type $A$ individuals. It is similar in spirit to the results above, showing that unless populations meet stringent requirements of similarity in utility distributions between groups (in this case instantiated by the fractions of the two types in each population), the versions of fairness characterized by comparing with the min versus the max cannot be simultaneously satisfied. 

\subsection{Multiplicative versus Additive Normalization}

\Improvement and \gainovermin fairness aim at capturing a similar fairness concept: groups receive on average the same increase in utility relative to assigning the least beneficial service. Both fairness metrics differ only by whether the normalization relative to the lowest utility that an individual can derive from the overall intervention is additive or multiplicative. In this section, we show that even the choice of normalization generates inherent fairness trade-offs. 

We consider the following stylized framework (SF2):
\begin{itemize}
    \item There are two types of individuals: type $C$ for which $u^{\min}$ takes a low value $\underline{u}$; and type $D$ for which  $u^{\min}$ takes a larger value $\overline{u} > \underline{u}$. 
    \item Conditional on $u^{\min}$, the distribution of utility is similar across programs and types.
\end{itemize}
Although stylized, both assumptions allow us to characterize the heterogeneity across groups by differences in their distribution over $u^{\min}$. Let $p_{s}$ denote the fraction of type $C$ for group $S=s$. Differences in $p_{s}$ across groups imply differences in the distribution of utility $P(U|S)$ within each group, even if the conditional distribution $P(U|U^{\min})$ is similar across types.

\begin{thm}
\label{thm: 2}
In the stylized framework (SF2) with types $C$ and $D$, a policy $\mathbf{a}$ satisfies both \improvement fairness and \gainovermin fairness for group $S=0$ and $S=1$ if and only if one of the following conditions holds:
\begin{itemize}
\item $p_{0}=p_{1}$;
\item the policy $\mathbf{a}$ assigns the least beneficial program to everyone (i.e. $\mathbf{a}.\mathbf{u}=u^{\min}$). 
\end{itemize}
\end{thm}

\begin{proof}
Let $\underline{\beta}$ denote $E[\mathbf{a}.\mathbf{u}|U^{min}=\underline{u}]$ and $\overline{\beta}$ denote $E[\mathbf{a}.\mathbf{u}|U^{min}=\overline{u}]$. Then, we write differences in \improvement as 
\begin{equation}
    \begin{split}
        \Delta I(\mathbf{a})& =\left\{p_{1}\underline{\beta}+(1-p_{1})\overline{\beta}-p_{1}\underline{u}-(1-p_{1})\overline{u}\right\}-\left\{p_{0}\underline{\beta} + (1 - p_{0})\overline{\beta}-p_{0}\underline{u}-(1-p_{0})\overline{u}\right\}   \\
        & =(p_{1}-p_{0})(\underline{\beta}-\overline{\beta} + \underline{u} - \underline{u})
    \end{split}
\end{equation}
and differences in \gainovermin as 
\begin{equation}
\Delta G(\mathbf{a}) = \left\{p_{1}\frac{\underline{\beta}}{\underline{u}} + (1 - p_{1}) \frac{\overline{\beta}}{\overline{u}}\right\}  -  \left\{p_{0}\frac{\underline{\beta}}{\underline{u}} + (1 - p_{0}) \frac{\overline{\beta}}{\overline{u}}\right\}=
    \left(p_{1}-p_{0}\right)\left(\frac{\underline{\beta}}{\underline{u}}-\frac{\overline{\beta}}{\underline{u}}\right)
\end{equation}
Therefore, \improvement fairness and \gainovermin are equivalent to
$(p_{0}-p_{1})(\underline{\beta}-\overline{\beta} + \underline{u} - \underline{u}) = 0$ and $\left(p_{0}-p_{1}\right)\left(\frac{\underline{\beta}}{\underline{u}}-\frac{\overline{\beta}}{\underline{u}}\right) = 0
$. If $p_{0}\neq p_{1}$, \improvement and \gainovermin fairness imply $\underline{\beta}=\frac{\underline{u}}{\overline{u}}\overline{\beta}$ and $\overline{\beta}= \overline{u}$. The latter equality leads to $\mathbf{a}.\mathbf{u}=\overline{u}$ if $u^{\min}=\overline{u}$ and the former equality leads to $\mathbf{a}.\mathbf{u}=\underline{u}$ if $u^{\min}=\underline{u}$.
\end{proof}
Theorem \ref{thm: 2} demonstrates a simple, yet general, setting where \improvement fairness and \gainovermin fairness cannot be obtained simultaneously unless either the distribution of utilities are the same across groups ($p_{0}=p_{1}$) or the policy does not create any utility improvement relative to $U^{\min}$.

\section{Simulations With Utilitarian and Random Allocations}\label{sec:exp}

Thus far, we have not needed to define an allocation policy explicitly, since we were focused on existence results. In this section, we consider two natural allocation policies -- utilitarian (maximizing the sum of utilities of all agents) and random. Both must respect capacity constraints. We simulate a simple environment with two groups and three services. In one setting, members of the two groups have different mean utilities from receiving the three services, while the variances are the same. In the second, members of the two groups have the same mean utilities from receiving the three services, but different variances. We are interested in understanding (1) how the different fairness measures behave in these two settings; (2) the role played by utilitarian objectives in the assignment problem.

In our setting, there are three ($k=1, 2, 3$) services with fixed capacities ($c_{1}=c_{2}=c_{3}=1000$) and $3000$ applicants divided into two groups of equal size: \emph{group 0} and \emph{group 1}. We sample individual utilities for service $k$ from a normal distribution with mean $\mu_{sk}$ and standard deviation $\sigma_{sk}$, where $s=0$ for \emph{group 0} and $s=1$ for \emph{group 1}.

\subsection{Groups with Different Means}
\begin{figure}
\subfloat[Distribution of $\Delta U$ ]{\includegraphics[width=.295\textwidth]{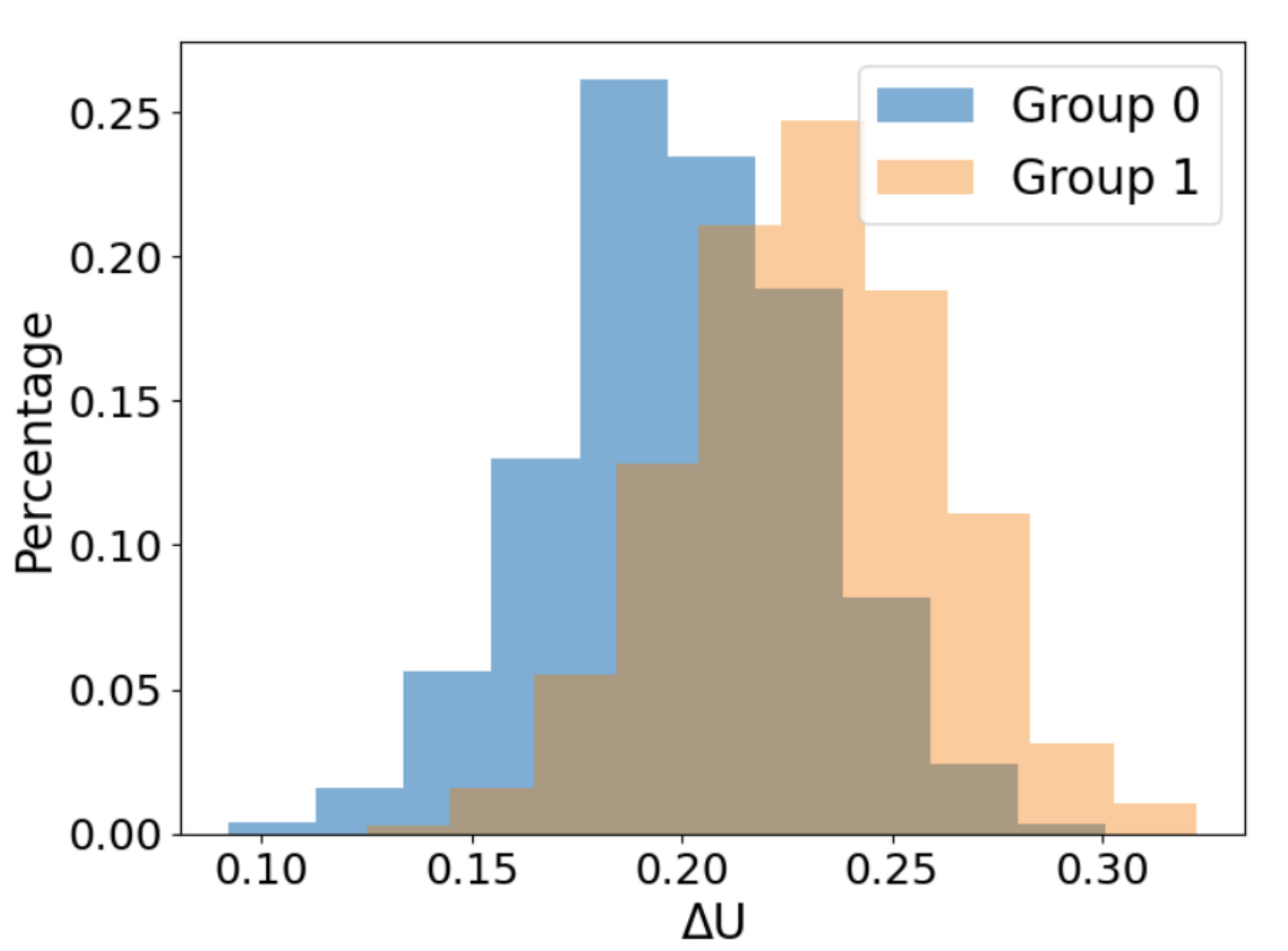}}
\subfloat[\Improvement vs. \Regret]{\includegraphics[width=.36\textwidth]{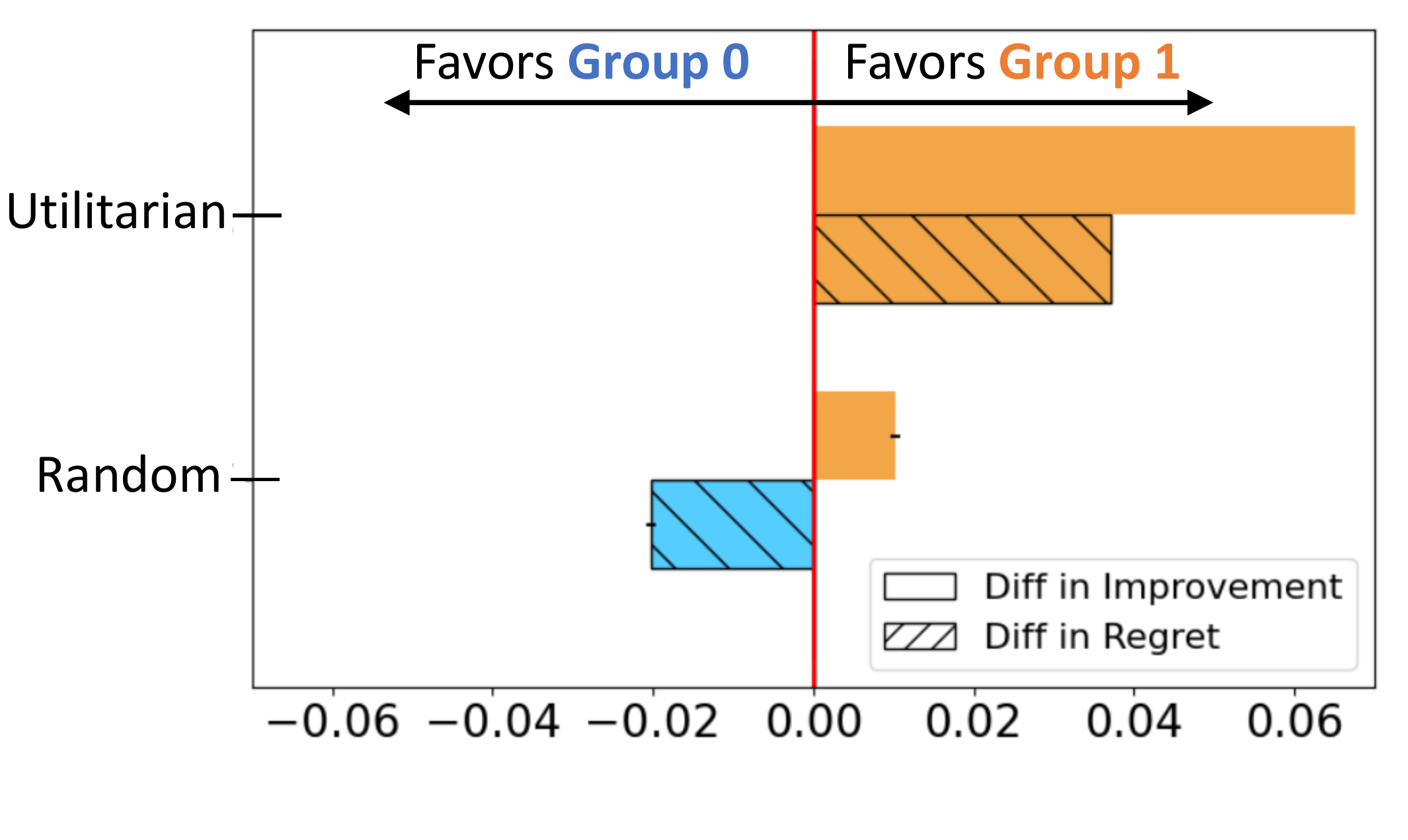}\label{fig:imp_reg_case1}}
\hspace{0.05em}
\subfloat[\Gainovermin vs. \Equitability]{\includegraphics[width=.305\textwidth]{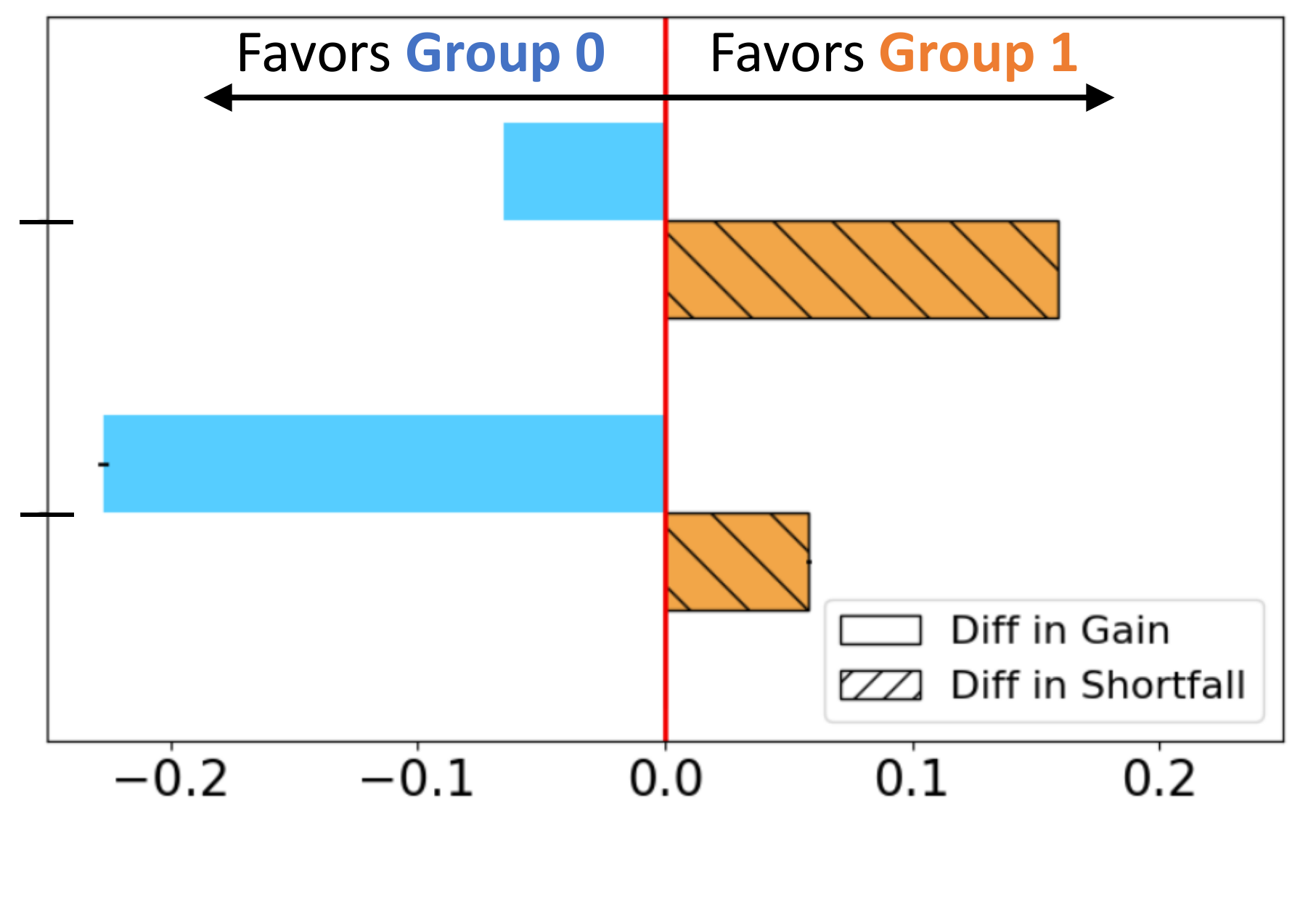}}
\caption{Simulation results when groups have different mean of utilities. Panel~(a) shows the distribution of the maximum utility gains  $\Delta U = U^{\max}-U^{\min}$ for \emph{group 0} (blue), and \emph{group 1} (orange). Panel~(b) shows the differences in improvement and regret, Panel~(c) shows the differences in gain and shortfall. Error bars show the 95\% confidence interval of each fairness metric over 100 instantiations of the random allocation.}
\label{fig:diff_mean_same_std}
\end{figure} 

In this set of simulations, we study the behavior of fairness measures when individual utilities are sampled from group-dependent distributions. The groups have different sample means $\mu$ but the same variances $\sigma^{2}$. For
 \emph{group 0}, the means of the three services are $\mu_{01}=0.2, \mu_{02}=0.3, \text{and}, \mu_{03}=0.4$. For \emph{group 1}, the means are $\mu_{11}=0.4, \mu_{12}=0.5,  \text{and}, \mu_{13}=0.63$
The variances of the three services for both groups are equal, $\sigma_{01}^{2} = \sigma_{11}^{2}= \num{1e-4}$
, $\sigma_{02}^{2} = \sigma_{12}^{2}=\num{4e-4}$, and, $\sigma_{03}^{2} = \sigma_{13}^{2}=\num{9e-4}$. Individuals in \emph{group 1} have on average higher utilities for all services. 

As pointed out in section~\ref{subsec:imp_vs_regret}, we observe in Figure \ref{fig:diff_mean_same_std} that the difference in $\Delta U$ leads to a trade-off between the \improvement and \regret fairness metrics. 
Figure \ref{fig:diff_mean_same_std} shows that even for a random assignment, different metrics lead to conflicting fairness assessment. 
The \improvement fairness metric favors the group with higher mean $\Delta U$ (\emph{group 1}), and \regret favors the groups with lower mean $\Delta U$ (\emph{group 0}). To complicate fairness assessment further, switching from additive to multiplicative normalization reverses which group is favored.

Moreover, 
the utilitarian allocation appears to favor \emph{group 1} according to \improvement, \regret and \gainovermin, but favors \emph{group 0} in terms of \equitability. These results confirm in a simulated environment that utility normalization has profound implications on how we assess the fairness of an allocation.

\subsection{Groups with Equal Means and Different Variances}

\begin{figure}

\subfloat[Distribution of $\Delta U$ ]{\includegraphics[width=.3\textwidth]{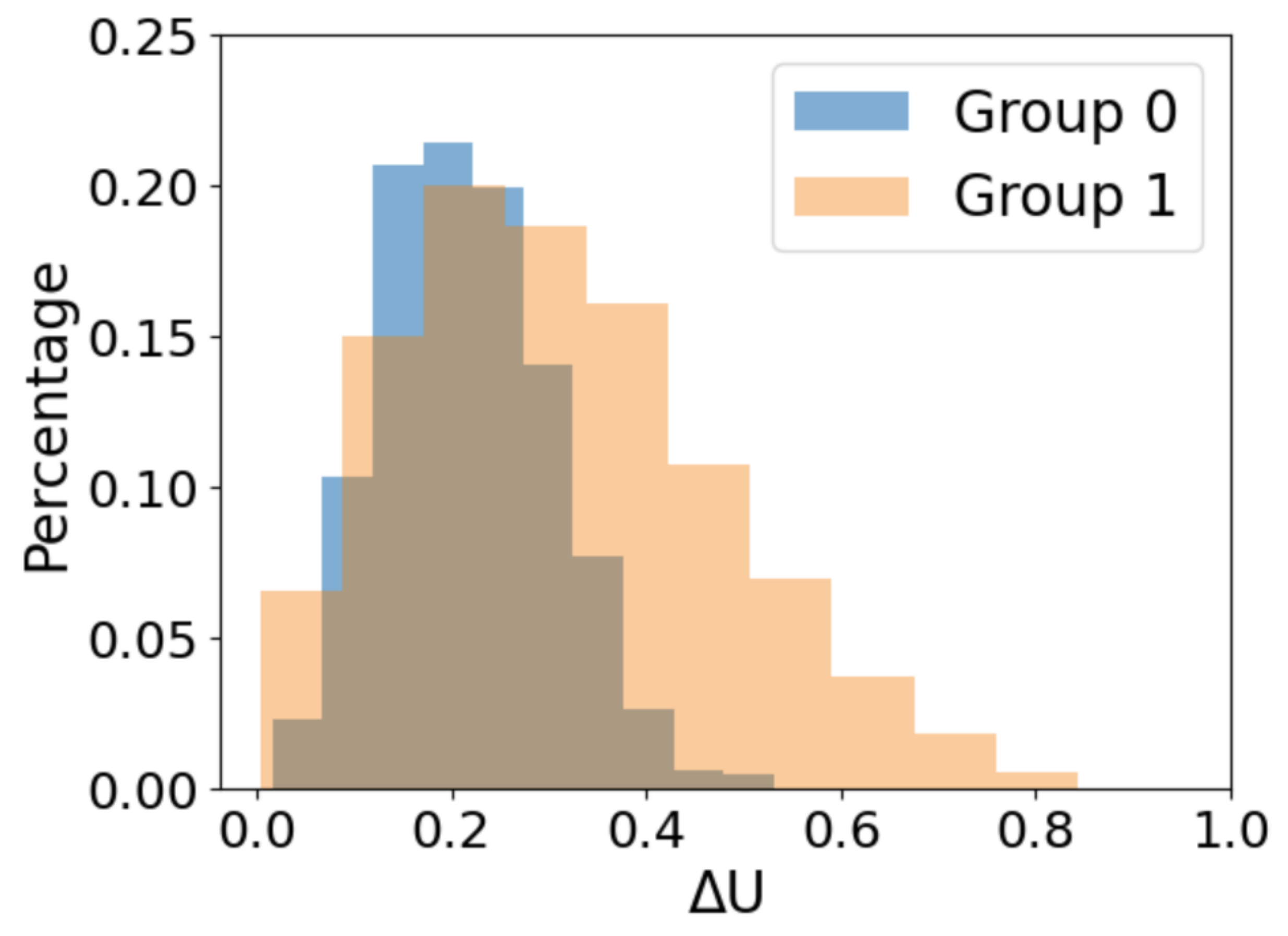}}
\subfloat[\Improvement vs. \Regret (Utilitarian)]{\includegraphics[width=.329\textwidth]{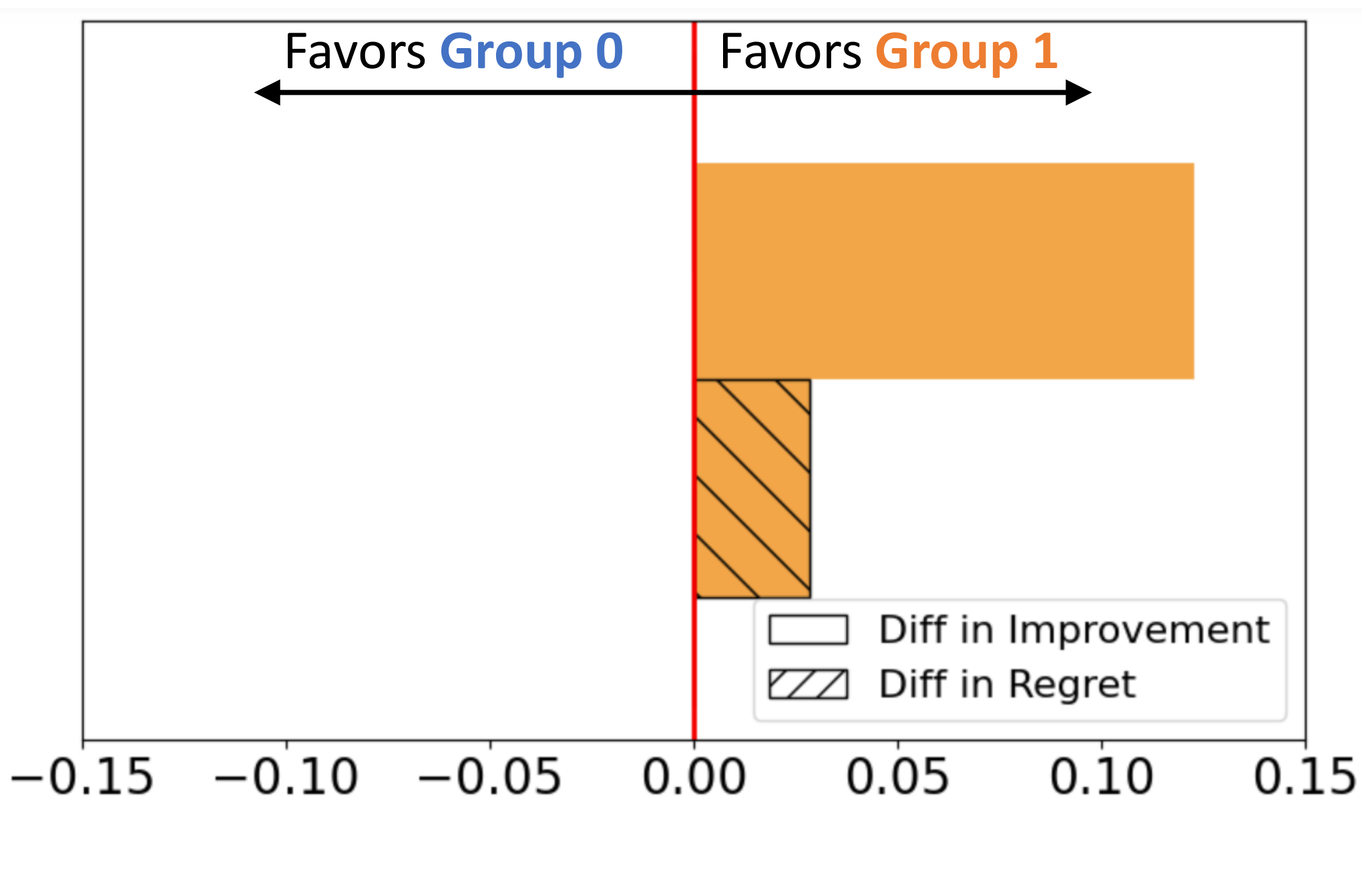}\label{fig:imp_reg_case2}}
\subfloat[ \Gainovermin vs. \Equitability (Utilitarian)]{\includegraphics[width=.33\textwidth]{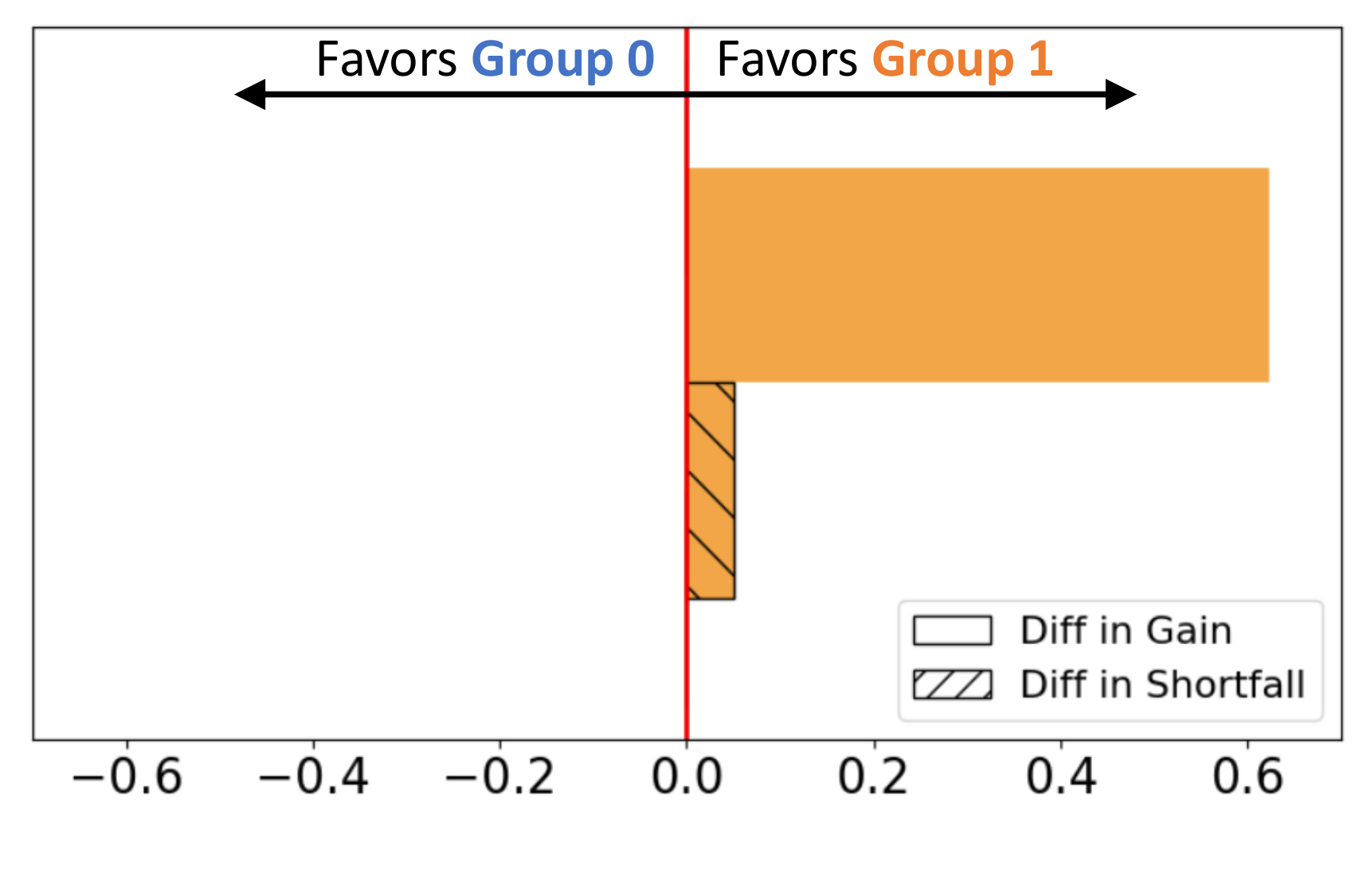}}\\
      \caption{Simulation results when groups have the same mean utilities for the services, but different variances. Panel~(a) shows the distribution of the maximum utility gains  $\Delta U = U^{\max}-U^{\min}$ for \emph{group 0} (blue), and \emph{group 1} (orange). Panel~(b) shows the differences in improvement and regret, Panel~(c) shows the differences in gain and shortfall. 
      Group 1 is favored strongly by all the fairness measures when allocations are utilitarian.}
      \label{fig:same_mean_diff_std}
\end{figure} 

In our second set of simulations, we study the effects of groups having similar means but different variances, a situation that is commonly discussed, for instance in the context of gender differences in student performance~\cite{baye2016gender}. In this case, we hypothesize that the higher variance group is likely to be favored by utilitarian allocations. For both groups, the means for the three services are equal,  $\mu_{01}=\mu_{11} = 0.4$, $\mu_{02}=\mu_{12} = 0.5$, and $\mu_{03}=\mu_{13} = 0.6$. 
For \emph{group 0}, the variances for the three interventions are set to $\sigma_{01}^{2} = \num{9e-5}$, $\sigma_{02}^{2} = \num{2e-3}$, $\sigma_{03}^{2} = \num{1e-2}$, while for \emph{group 1},
the variances for the three interventions are set to
$\sigma_{11}^{2} = \num{9e-3}$, $\sigma_{12}^{2} = \num{2e-2}$, 
$\sigma_{13}^{2} = \num{3e-2}$. Thus, \emph{group 0} has lower variance. 

Our results in Figure \ref{fig:same_mean_diff_std} show that, as hypothesized, the group with larger variance (group 1) is indeed favored according to all fairness metrics. When maximizing the sum of utilities, it is optimal to assign their best services to individuals with utilities in the tail of the distribution. We find that a larger fraction (65\%) of individuals in \emph{group 1} than in \emph{group 0} (46\%) receive the service that maximizes their utility.

We leave it for future research to investigate further the role of variance on the fairness properties of a utilitarian allocation.

\section{Fairness Trade-offs in Homeless Service Delivery}

Our theoretical analysis suggests that heterogeneity in service responses across groups drives fairness metrics in opposite directions. In this section, we investigate whether the fairness tradeoffs emerge in the capacitated assignment of homeless services across several sub-populations. We hypothesize that if sociodemographic group differences exist in the utilities received from allocations (and in particular, between the differences in the best versus worst allocations), then we should see tradeoffs between \improvement versus \regret fairness,  \equitability versus \gainovermin, and \improvement versus \gainovermin. We provide evidence for both the antecedent (\emph{heterogeneity in responses across groups}) and the consequent (\emph{inherent fairness trade-offs between groups}).

\subsection{Background}
Homelessness represents a socioeconomic and public health challenge for many communities in the United States. Approximately $1.5$ million people experience homelessness for at least one night every year \cite{henry2020ahar, fisher2018homelessness}. Homelessness has short- and longer-term implications on health, employment, and crime \cite{fowler2019solving, khadduri2010costs, cohen2020housing}. Guided by federal policies, communities offer an array of services for households lacking stable and permanent living accommodations. We study three main homeless services: Transitional Housing (TH); Rapid Rehousing (RRH) and Emergency Shelter (ES). Transitional Housing provides accommodation for up to 24 months with comprehensive case management to address barriers toward stable housing, such as substance abuse and issues related to behavioral health. Rapid Rehousing offers access to rental units for six months without intensive case management. Emergency Shelter provides a bed to sleep at night for no more than one or two months. On a daily basis, caseworkers assign homeless households seeking assistance to an available service, reserving the most intensive TH for those with greater needs.

\subsection{Data}
\label{sec:data}
Our main dataset is based on estimated probabilities of households re-entering homelessness services within two years after initial receipt of services. This data, collected by \cite{kube2019fair} is publicly available.\footnote{\url{https://github.com/amandakube/Allocating-Homelessness-Interventions---Counterfactual-Predictions}} The estimates are based on applying a machine learning model (BART \cite{hill2011bayesian}) to administrative records that tracked service provision in a metropolitan area from 2007 through 2014. Service providers collected demographic and household characteristics upon entry into the system, and data capture the intervention assigned and whether households subsequently requested additional assistance \cite{kube2019fair}. The model estimates counterfactual probabilities $p_{ik}$ of a household $i$ to re-enter the homeless system within 2 years given the assignment of a specific service $k$, where $k\in \{TH, RRH, ES\}$. The original data also tracks responses to homelessness prevention -- time-limited monetary assistance that differs from the other three interventions that allocate actual bed space. Given that the constraints on homelessness prevention are different, we focus here only on households that needed actual bed space (and were therefore not eligible to receive prevention services). Therefore, our final data contains $3,375$ households and they received either TH, RRH, or ES. 

We compute the utility of service $k$ to individual $i$ as $u_{ik}=1-p_{ik}$. 
We obtained from Kube et al. additional sociodemographic characteristics
for each household, including race, gender, age, disability status,
presence of spouse and/or children, and household size. 

We define a series of sociodemographic groups and intersectional identities expected to exhibit substantial heterogeneity in responses to homeless services. First, households with disabilities are considered more vulnerable, and prior research shows that more vulnerable households do best with more intensive services ~\cite{aubry2020,munthe-kaas2018}. Therefore, we expect households with disabilities to benefit more from TH and less from ES than the rest of the population. Second, families with children under the age of 18 experience homelessness due to socioeconomic reasons rather than disability and vulnerability, and thus, we anticipate families will respond better to rapid rehousing than more intensive TH~\cite{cunningham2015rapid, fertig2008homelessness, rog2007characteristics}. Third, we examine the intersection between gender and family status, assuming that single female households without children do better in TH compared with single female-headed families with children, who are more likely to benefit from RRH. Fourth, we look within households headed by youth aged 18 to 24 years to compare disability status (versus no disability) and family status (children versus no children), hypothesizing that those with disabilities benefit more from TH and families with children from RRH \cite{morton2020}. Lastly, given the over-representation in homelessness of minorities and especially Black households, we test how race affects homeless service utilities \cite{henry2020ahar}. Prior research suggests the causes of homelessness vary for White people, who more likely experience disabilities, versus Black people, who experience greater housing discrimination and marginalization ~\cite{jones2016does}. Moreover, race intersects with gender (males vs females) and family status (with children versus without children) in ways that could drive variation in homeless service outcomes. 

\subsection{Heterogeneity across Demographic Groups}
In this section, we document heterogeneity in the distributions of utility across various sociodemographic groups. For each household, we compute the difference $\Delta U$ between its best and worst utility. 

Figure \ref{fig:delta_u_het} shows heterogeneity in response to homeless services across 
households with and without reported disabilities; with and without children. The distribution of $\Delta U$ for households with a disability skews to the right (panel a)); assigning the best service to a disabled client has a larger impact in terms of the probability to re-enter the homeless system than assigning a client without a disability to its most beneficial service. The difference in the means of the distributions is statistically significant with a t-statistic of 8.5 and p-value infinitesimally small. This finding aligns with prior research that shows vulnerable households do best with more intensive services \cite{aubry2020,munthe-kaas2018}. The distribution of $\Delta U$ for families without children skews strongly to the right compared with to households with children (panel b)). The mean of $\Delta U$ for households without children is 0.07, while it is only 0.04 for household with children. The difference is statistically significant with a t-statistic of 29.0 and a p-value  infinitesimally small. This result illustrates how families with children differ in their responses to housing assistance compared to homeless individuals.  

\begin{figure}
    \centering
    \includegraphics[width=0.95\textwidth]{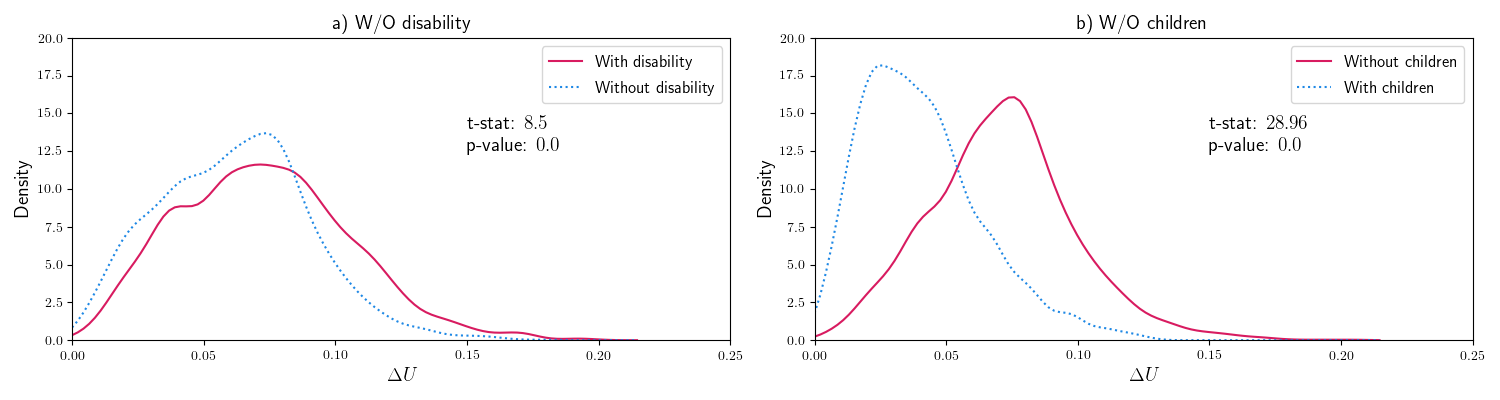}
    \caption{Distribution of the maximum utility gain $\Delta U$ that individuals can derive from the homeless system across various demographic groups. We  obtain the probability density function of $\Delta U = U^{max}-U^{min}$ via Gaussian kernel density estimation with a bandwidth of $0.2$. Differences in probability density functions between households with and without disability (Panel a)) and with and without spouse (Panel b)) illustrate heterogeneous responses to housing assistance.}
    \label{fig:delta_u_het}
\end{figure}

In Figure \ref{fig:delta_u_hom}, we look at intersectional sociodemographic groups. We find in panel c) that the impact of different homeless services for a single female depends strongly on whether there are children in the household. Similarly, youth with and without disability respond differently to homeless services (panel d)). For both intersections, the difference in means is statistically significant with a t-statistic equal to 25.7 for single female versus single mother and to 5.1 for youth with a disability versus youth without a disability. 

Figure \ref{fig:delta_u_inter} explores differential responses to housing assistance by race and shows substantial differences in the distribution of $\Delta U$ between Black and White males (Panel g)). 
Black homeless populations may on average benefit more from more intensive homeless services. Prior research \cite{jones2016does} suggests that social discrimination and socio-economic disadvantage could increase the risk for homelessness among populations with perceived Black background and that housing assistance could mitigate some of these vulnerabilities. 

\begin{figure}
    \centering
    \includegraphics[width=0.95\textwidth]{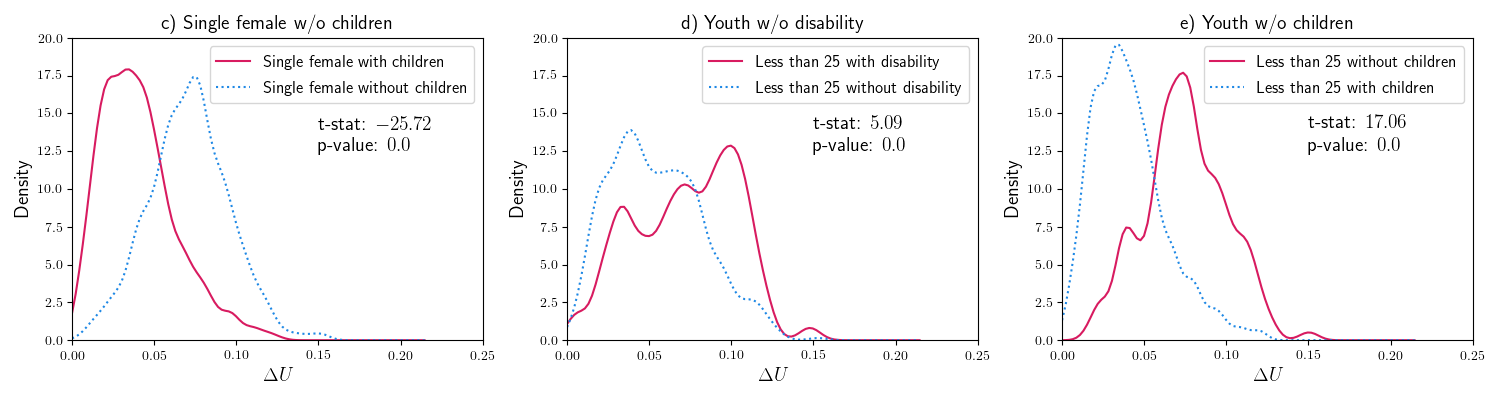}
    \caption{Same as Figure \ref{fig:delta_u_het} but for intersection groups single female with and without children (Panel c)); youth under 25 with and without disability (Panel d)); and, youth under 25 with and without children (Panel e)). }
    \label{fig:delta_u_hom}
\end{figure}

\begin{figure}
    \centering
    \includegraphics[width=0.95\textwidth]{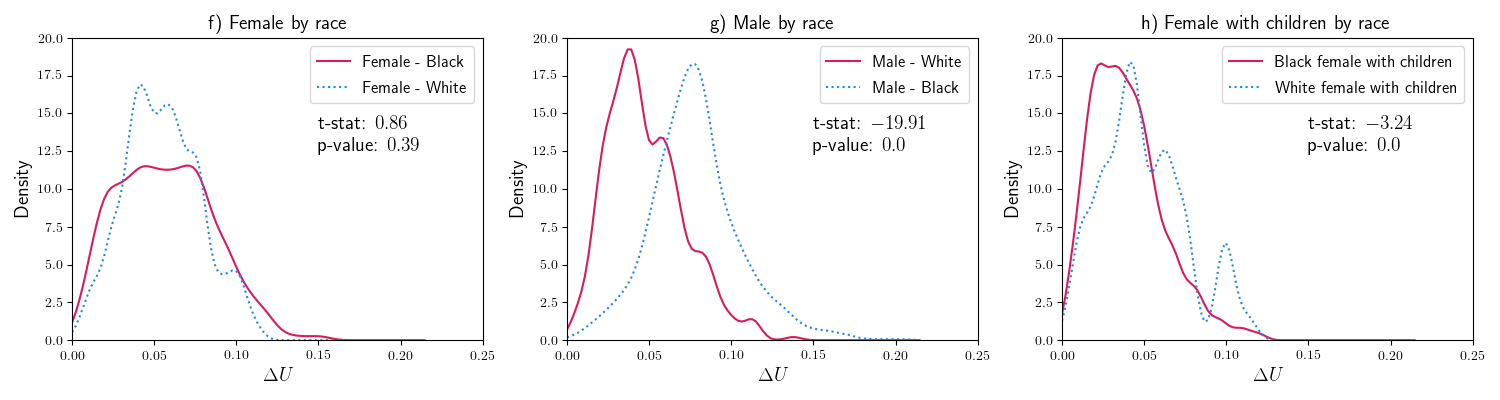}
    \caption{Same as Figure \ref{fig:delta_u_het} but for groups defined by perceived racial background.}
    \label{fig:delta_u_inter}
\end{figure}

Results from Figures \ref{fig:delta_u_het}, \ref{fig:delta_u_hom} and \ref{fig:delta_u_inter} 
suggest that heterogeneity in utility pervades sociodemographic groups. Table \ref{tab: dist_serv} explains some of this heterogeneity by identifying which of the three services (TH, RRH and ES) benefits the most households within each group. 
For the homeless population studied in this paper, TH is the most preferred service for $68\%$ of the population, followed by RRH ($27\%$) and ES ($5\%$). We find that this preference for more intensive care is exacerbated for households with disability ($73\%$ prefer TH), which is in line with prior findings that most vulnerable populations benefit from more integrated care. The preferences of households with a disability toward TH contrasts with the preferences of families with children toward RRH: $67\%$ of households with children benefit the most from RRH, while TH is the best service for only $16\%$ of families. This observation holds true for all intersectional groups that include children and could explain differences between males and females, 
since females are more likely to live with children than males. On the other hand, regardless of gender, the most beneficial program is more likely to be TH for the Black homeless population: TH is the most beneficial service for $46\%$ of Black females but only for $34\%$ of White females;  and, for $95\%$ of Black males but only for $80\%$ of White males.

\begin{table}[]
\centering
\caption{Distribution of services that deliver to each household the highest utility across demographic groups. This shows the fraction of households in each demographic group for which ES, TH or RRH leads to the lowest probability to re-enter the homeless system. }
\begin{tabular}{llllllll}
 \toprule
                                & TH   & RRH  & ES  &  & TH   & RRH  & ES \\
                \midrule               
All                             & 0.68 & 0.27 & 0.05 \\
\midrule
With disability                 & 0.73 & 0.23 & 0.03 & Without disability              & 0.66 & 0.28 & 0.06\\
\midrule 
Without children                & 0.85 & 0.14 & 0.01 & With children                   & 0.16 & 0.67 & 0.17 \\
\midrule 
Single female with children     & 0.15 & 0.7  & 0.15 & Single female without children  & 0.7  & 0.3  & 0.01 \\
\midrule 
Less than 25 with disability    & 0.62 & 0.37 & 0.01 & Less than 25 without disability & 0.47 & 0.49 & 0.05 \\
\midrule
Less than 25 without children   & 0.83 & 0.17 & 0.0  & Less than 25 with children      & 0.19 & 0.73 & 0.08 \\
\midrule
Female - Black       & 0.46 & 0.46 & 0.08 & Female - White                  & 0.34 & 0.62 & 0.04 \\
Male - Black                    &0.95  & 0.02 & 0.03 & Male - White                    & 0.8  & 0.14 & 0.06 \\
\bottomrule
\end{tabular}
\label{tab: dist_serv}
\end{table}

\subsection{Fairness Trade-Offs in the Observed Allocation of Homeless Services}
Our theory suggests that heterogeneity in the distribution of the maximum gain $\Delta U$ for housing assistance would drive fairness metrics in opposite directions: (i) there exist assignments of homeless services with conflicting fairness assessment depending on choosing \improvement, \regret, \gainovermin or \equitability as the fairness metric (Theorem~\ref{cor: 2}); (ii) assignments that satisfy improvement fairness could violate regret fairness and vice-versa (Theorem~\ref{thm: trade-off}). 
Since we observe substantial heterogeneity among the sociodemographic and intersectional groups presented in section 5.3, we know by Theorem~\ref{cor: 2} that ambiguous fairness assessments can arise for some policies. However, Theorem \ref{cor: 2} is not constructive and does not tell whether such policies are realistic in the context of homeless services delivery. Here we test whether the observed assignment as reported in the administrative records is subject to contradictory fairness assessments depending on the choice of the fairness metric. 

Figure~\ref{fig:trade-off-obs} (Panel a)) plots the difference in improvement $\Delta I$ and the negative of difference in regret $-\Delta R$, so that positive values indicate that the policy favors group $S=1$, while negative values mean the policy favors group $S=0$. According to the \improvement metric, the observed assignment favors households without children, while according to \regret, it favors households with children: $\Delta I$ is equal to $-0.013$, while $-\Delta R$ is equal to $0.016$. A similar ambiguity emerges for households with and without disability. Moreover, choosing \improvement over \regret flips the conclusion on whether the observed assignment is unfair to Black males relative to White males: Black males derive higher utility gains according to \improvement ($\Delta I = -0.02$) but lower utility gains according to regret ($-\Delta R = 0.009$). 
The results provide empirical evidence that policies that lead to contradictory fairness assessment in Theorem \ref{cor: 2} are not just theoretical oddities, but do occur in real world applications. Although we do not prove a counterpart of 
Theorem~\ref{cor: 2} for \equitability versus \gainovermin, we find empirically that similar trade-offs do, in fact, occur (Figure \ref{fig:trade-off-obs}, Panel b)). Moreover, in Figure~\ref{fig:trade-off-obs}, we find one pairwise comparison, youth with a disability versus youth without a disability, for which the observed policy satisfies \improvement fairness. This instance of \improvement fairness allows us to test whether Theorem~\ref{thm: trade-off} holds here. We find that the policy does not satisfy \regret fairness, which is consistent with the heterogeneity in $\Delta U$ found in section 5.3 between youth with a disability versus youth without a disability. 

\begin{figure}
    \centering
    \includegraphics[width=1.0\textwidth]{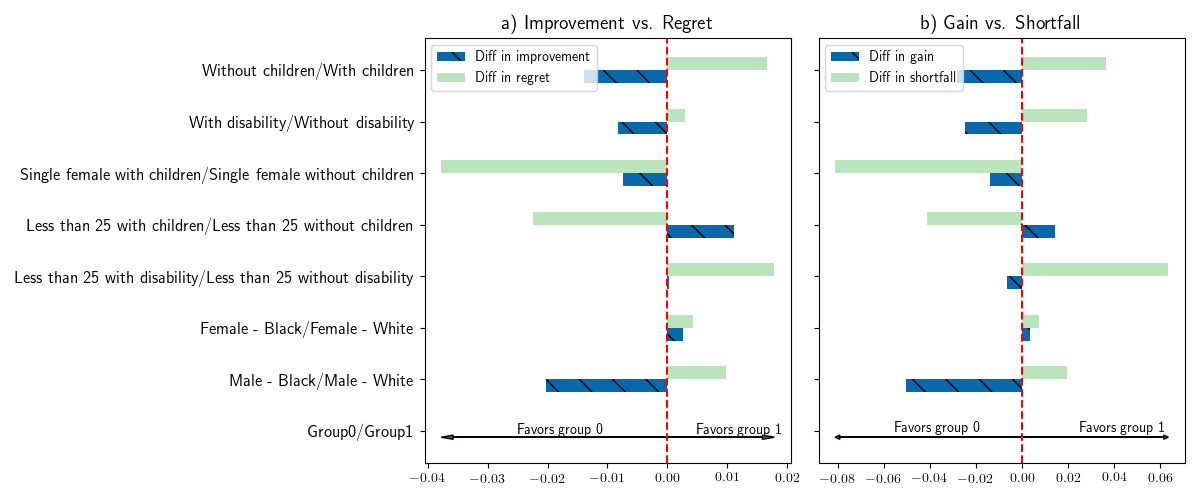}
    \caption{Fairness trade-off in the observed assignment of homeless services. This compares which demographic group is favored by the assignment depending on the fairness metric. Trade-offs occurs when \improvement favors one group and \regret the other one (left panel) or when \equitability favors one group and \gainovermin the other one (right panel).}
    \label{fig:trade-off-obs}
\end{figure}
\section{Conclusion}
How do we judge whether an approach to allocation of scarce societal resources is fair for different sociodemographic groups of public concern? The problem lies at the intersection of recent work in fair machine learning and a long history of work from economics, social choice, and algorithmic game theory on fair division. It also brings into question concerns of local justice~\cite{elster1992local}, which studies how individuals are prioritized in the allocation of scarce resources by local institutions. The key point we make in this paper is that \emph{baselines matter when we measure outcomes for different groups}. The exact same allocation may favor one group over another when assessed against the baseline intervention of doing nothing, but the group it favors could invert when measured against the baseline of giving each group the best intervention it could get in a scenario with no resource constraints. The social objective being optimized also can drive fairness results -- for example, utilitarian allocations will typically favor groups with higher variance in utilities across different types of services, even if the means are the same. 

Our results are more than theoretical. We show that the pattern arises in homeless service delivery, where outcomes vary by and within sociodemographic groups. For instance, returns to homelessness vary by service allocation more for households without children compared to families with children. Naive policy applications that fail to consider baseline variation may negatively impact some groups. Aiming to reduce overall homelessness, for example, by prioritizing households without children for intensive services disproportionately excludes households with children from receiving their best service, whereas an alternative policy that matches households with children to their best service fails to reduce overall homelessness. The data illustrate similar fairness tradeoffs across intersecting sociodemographic groups, including disability status, gender, age, and race.  Failing to consider carefully the underlying distributions and metrics for success threatens counterproductive policy initiatives. Current national advocacy to reduce veteran and chronic homelessness to zero ask communities to shift resources in ways that may undermine other goals~\cite{builtforzero}. Moreover, federal and local policies simultaneously strive for system efficiency and equity, which prove antithetical in many contexts~\cite{fowler2019solving}. Our findings raise serious questions for institutions when designing homeless policies and social policy more generally.

\bibliographystyle{unsrt}  
\bibliography{sample-base}

\end{document}